\documentclass[11pt,oneside]{article}

\usepackage[T1]{fontenc}
\usepackage[margin=1in]{geometry}
\usepackage{mathpazo}
\usepackage{xcolor}

\newcommand{\red}[1]{\textcolor{red}{#1}}

\newif\ifshowtodos
\showtodostrue

\usepackage{booktabs}
\usepackage{caption}
\usepackage{graphicx}

\let\oldincludegraphics\includegraphics%
\renewcommand{\includegraphics}[2][]{\IfFileExists{#2}{\oldincludegraphics[#1]{#2}}{\red{[FILE NOT FOUND]}}}

\usepackage{amsmath}
\usepackage{amssymb}
\usepackage{amsthm}

\DeclareMathOperator{\E}{\mathbb{E}}
\DeclareMathOperator{\Var}{\mathbb{V}}

\newcommand{\cost}{C}
\newcommand{\der}{\mathrm{d}}
\newcommand{\derfrac}[2]{\frac{\der #1}{\der #2}}
\newcommand{\Hcal}{\mathcal{H}}
\newcommand{\N}{\mathcal{T}}
\newcommand{\Ncal}{\mathcal{N}}
\newcommand{\parfrac}[2]{\frac{\partial #1}{\partial #2}}
\newcommand{\pv}{\Psi}
\newcommand{\R}{\mathbb{R}}
\newcommand{\vm}{\sigma^2}
\newcommand{\vz}{\sigma_0^2}
\renewcommand{\epsilon}{\varepsilon}

\let\oldleft\left
\let\oldright\right
\renewcommand{\left}{\mathopen{}\mathclose\bgroup\oldleft}
\renewcommand{\right}{\aftergroup\egroup\oldright}

\newtheorem{lemma}{Lemma}
\newtheorem{proposition}{Proposition}

\newif\ifbodyproofs
\bodyproofstrue

\usepackage{needspace}
\AtBeginEnvironment{lemma}{\Needspace*{6\baselineskip}}
\AtBeginEnvironment{proposition}{\Needspace*{6\baselineskip}}
\AtBeginEnvironment{theorem}{\Needspace*{6\baselineskip}}

\usepackage{natbib}
\newcommand\citeapos[1]{\citeauthor{#1}'s (\citeyear{#1})}

\bodyproofsfalse

\title{Persistence, patience and costly information acquisition}
\author{%
Benjamin Davies\thanks{
Department of Economics, Stanford University; bldavies@stanford.edu.
}
}
\date{Draft version: \today}

\begin{document}

\thispagestyle{empty}

\maketitle

\begin{abstract}
    \noindent
    A forward-looking agent observes signals of a state that follows a Gaussian AR(1) process.
    He balances the cost of having imprecise beliefs with the cost of acquiring more precise signals.
    I characterize his optimal information acquisition policy, and analyze how his steady-state beliefs and costs depend on persistence (the AR(1) parameter) and patience (the agent's discount factor).
    Higher persistence has a non-monotone effect on belief precision and raises overall costs.
    Higher patience makes beliefs more precise and lowers overall costs.

    \vskip\baselineskip
    \noindent{\itshape Keywords}: AR(1) process, Bayesian learning, information acquisition, patience, persistence\par
    \noindent{\itshape JEL classification}: C61, D83
\end{abstract}

\clearpage
\setcounter{page}{1}
\section{Introduction}

Investors, central banks, and many other agents face a common challenge: learning about unknown states that change over time.
These agents must decide how much to learn given the cost of information and its gradual obsolescence.
How is this decision shaped by a state's persistence and the agent's patience?

To answer this question, I consider a forward-looking agent who observes signals of a payoff-relevant state that follows a Gaussian AR(1) process.
He balances the instrumental cost of having imprecise beliefs with the cost of acquiring more precise signals.
His optimal policy minimizes the present value of overall costs.
I characterize this policy, and analyze how steady-state beliefs and costs depend on persistence (the AR(1) parameter) and patience (the agent's discount factor).

Higher persistence has two opposing effects on steady-state beliefs:
\begin{enumerate}

    \item
    shocks propagate further into the future, making beliefs less precise;

    \item
    information becomes obsolete slower, so the agent acquires more, making beliefs more precise.

\end{enumerate}
I show that the first effect dominates the second if and only if persistence is high enough (Proposition~\ref{prop:steady-state-rho}).
However, higher persistence always leads to higher overall costs (Proposition~\ref{prop:steady-state-cost}).
This is because instrumental costs fall by less than acquisition costs rise.
The opposite is true when patience rises: the agent acquires more information, but so do his past selves, and the inherited information lowers his instrumental costs by more than acquisition costs rise.

I obtain these results in a tractable model with Gaussian signals, quadratic payoffs, and linear signal precision costs.
My Gaussian-quadratic setup is standard in the information acquisition literature \citep{Colombo-etal-2014-REStud,Hellwig-Veldkamp-2009-REStud,Myatt-Wallace-2012-REStud}.
Linear costs distinguish my model from the rational inattention literature, which typically assumes entropy-based costs \citep{Mackowiak-Wiederholt-2009-AER,Sims-2003-JME}.
Such costs are prior-dependent, complicating dynamic analyses because each signal changes future priors and, thus, future costs.
In contrast, linear costs are prior-independent and admit a closed-form solution to the agent's dynamic problem.%
\footnote{
Moreover, \cite{Pomatto-etal-2023-AER} show that linear costs uniquely satisfy a set of appealing axioms.
}

The closest papers to mine are \cite{Steiner-etal-2017-ECTA}, \cite{Weber-Nguyen-2018-EJOR}, \cite{Immorlica-etal-2021-LeibnizInt.Proc.Inform.LIPIcs}, and \cite{Barilla-2025-}.
All study optimal learning about a time-varying state.
\cite{Steiner-etal-2017-ECTA} assume entropy-based costs, while \cite{Weber-Nguyen-2018-EJOR} assume linear costs.
Neither paper considers the effects of persistence and patience.
\cite{Immorlica-etal-2021-LeibnizInt.Proc.Inform.LIPIcs} study precision allocation under a budget constraint, removing the trade-off between instrumental and acquisition costs that I analyze.
\citeapos{Barilla-2025-} model has a binary state and endogenous signal timing.
My focus on a continuous state and exogenous timing leads to different insights.

% \clearpage
\section{Model}
\label{sec:model}

\paragraph{Environment}

Time is discrete and indexed by~$t\in\N\equiv\{1,2,\ldots\}$.
The state follows a Gaussian AR(1) process~$(\theta_t)_{t\in\N}$ with autocorrelation~$\rho\in(0,1)$, initial value~$\theta_0\sim\Ncal(0,\vz)$, and independent shocks~$\eta_t\sim\Ncal(0,\vm)$:%
\footnote{
An early version of this paper \citep{Davies-2026-} focused on the limiting case when~$\rho\to1$.
}
\[ \theta_t=\rho\theta_{t-1}+\eta_t. \]
I assume~$\vz\ge0$ and~$\vm>0$.

\paragraph{Information}

At each time~$t$, the agent generates a Gaussian signal~$s_t$ of~$\theta_t$ with precision~$x_t\ge0$:
\[ s_t\mid\theta_t\sim\Ncal\left(\theta_t,\,\frac{1}{x_t}\right). \]
He observes the signal history~$\Hcal_t\equiv(s_1,\ldots,s_t)$ and uses it to form posterior beliefs about~$\theta_t$.

\paragraph{Prediction and posterior variances}

Let~$\Hcal_0=\{\}$, and let~$P_t\equiv\Var(\theta_t\mid\Hcal_{t-1})$ and~$V_t\equiv\Var(\theta_t\mid\Hcal_t)$ be the posterior variances of~$\theta_t$ before and after observing~$s_t$.
Well-known results for Gaussian random variables imply
\[ V_t=\left(\frac{1}{P_t}+x_t\right)^{-1}. \]
So given~$P_t>0$, every precision~$x_t\ge0$ corresponds to a unique posterior variance~$V_t\in(0,P_t]$.
I call~$P_t$ the ``prediction variance'' because it is the variance of the prediction of~$\theta_t$ from time~$(t-1)$.

\paragraph{Costs}

After observing~$\Hcal_t$, the agent takes an action~$a_t\in\R$ with cost~$(a_t-\theta_t)^2$.
He minimizes the expected cost by choosing~$a_t$ to equal the posterior mean~$\E[\theta_t\mid\Hcal_t]$.
This yields minimized cost
\[ \min_{a'\in\R}\E[(a_t-\theta_t)^2\mid\Hcal_t]=V_t. \]
Thus, the agent incurs ``instrumental costs'' from having imprecise beliefs.
He also incurs ``acquisition costs'' from generating precise signals.
His overall cost at time~$t$ equals
\[ \cost(V_t,P_t)\equiv V_t+c\left(\frac{1}{V_t}-\frac{1}{P_t}\right), \]
where~$c>0$ is the marginal cost of raising the precision~$x_t=1/V_t-1/P_t$ of the signal~$s_t$.
\iffalse
\footnote{
The cost function~$cx_t$ maps to a setting where the agent can buy any non-negative real number of iid signals at a constant per-unit price.
For example, if each signal has precision~$1/\sigma_\epsilon^2\ge0$ and price~$\kappa>0$, then~$m\ge0$ signals have combined precision~$x\equiv m/\sigma_\epsilon^2$ and cost~$\kappa m=cx$ with~$c\equiv\kappa\sigma_\epsilon^2$.
}
\fi

\paragraph{Choices}

The agent has discount factor~$\delta\in[0,1)$.
At each time~$t$, he chooses the posterior variance~$V_t$ (by choosing a corresponding precision~$x_t$) that minimizes the present value
\[ \cost(V_t,P_t)+\sum_{\tau=1}^\infty\delta^\tau\cost(V_{t+\tau},P_{t+\tau}) \]
of his overall costs, given
the variances~$(V_\tau)_{\tau<t}$ he chose at past times and
the variances~$(V_\tau)_{\tau>t}$ he anticipates choosing at future times.

Since~$P_t$ is a sufficient statistic for~$\Hcal_{t-1}$ and since the AR(1) structure implies~$P_{t+1}=\rho^2V_t+\vm$, the minimized present value~$\pv(P_t)$ depends on~$P_t$ only and satisfies the Bellman equation
\begin{equation}
    \label{eq:bellman}
    \pv(P_t)=\min_{V_t\in(0,P_t]}\left\{\cost(V_t,P_t)+\delta\pv(\rho^2V_t+\vm)\right\}.
\end{equation}
So the agent's time~$t$ problem is to choose, given~$P_t$, the posterior variance~$V_t$ that minimizes the sum of the current cost~$\cost(V_t,P_t)$ and continuation value~$\delta\pv(\rho^2V_t+\vm)$.

% \clearpage
\section{Optimal policy}
\label{sec:optimal-policy}

Lemma~\ref{lem:optimal-policy} characterizes the optimal posterior variances~$(V_t)_{t\in\N}$ in terms of the (endogenous) prediction variances~$(P_t)_{t\in\N}$.

\begin{lemma}
    \label{lem:optimal-policy}
    Let~$V^*$ be the unique positive solution to
    \begin{equation}
        \label{eq:foc}
        \frac{1}{(V^*)^2}-\frac{\delta\rho^2}{\left(\rho^2V^*+\vm\right)^2}=\frac{1}{c}.
    \end{equation}
    At each time~$t$, the agent optimally chooses~$V_t=\min\left\{P_t,V^*\right\}$ by generating a signal with precision
    \[ x_t=\max\left\{0,\frac{1}{V^*}-\frac{1}{P_t}\right\}. \]
\end{lemma}
\ifbodyproofs\subsection{Proof of Lemma~\ref{lem:optimal-policy}}

\begin{proof}[\unskip\nopunct]
    I first show~$V^*>0$ exists and is unique.
    Define the function~$f:(0,\infty)\to\R$ by
    \[ f(V)\equiv\frac{1}{V^2}-\frac{\delta\rho^2}{\left(\rho^2V+\vm\right)^2}. \]
    Then~$f(V)\to\infty$ as~$V\to0$ (because~$1/V^2\to\infty$ while the second term is bounded) and~$f(V)\to0$ as~$V\to\infty$.
    But~$f$ is continuous on its domain and~$0<1/c<\infty$, and so the intermediate value theorem implies there is at least one~$V>0$ such that~$f(V)=1/c$.
    Now~$f$ has derivative
    \[ f'(V)=-2\left(\frac{1}{V^3}-\frac{\delta\rho^4}{\left(\rho^2V+\vm\right)^3}\right). \]
    There are two cases to consider:
    \begin{enumerate}

        \item[(i)]
        Suppose~$\delta\le\rho^2$.
        Then for all~$V>0$ we have
        \[ \rho^2V+\vm>\delta^{1/3}\rho^{4/3}V, \]
        implying~$\delta\rho^4/\left(\rho^2V+\vm\right)^3<1/V^3$ and therefore~$f'(V)<0$.
        So~$f$ is strictly decreasing on its domain and hence~$f(V)=1/c$ has exactly one solution.

        \item[(ii)]
        Now suppose~$\delta>\rho^2$.
        Then~$f'(V_0)=0$ at a strictly positive point
        \[ V_0\equiv\frac{\vm}{\rho^{4/3}\left(\delta^{1/3}-\rho^{2/3}\right)}, \]
        with~$f'(V)<0$ if and only if~$V<V_0$.
        Moreover,
        \begin{align*}
            % \rho^2V+\vm
            % &= \delta^{1/3}\rho^{4/3}V_0 \\
            % \Rightarrow
            f(V_0)
            % &= \frac{1}{V_0^2}-\frac{\delta\rho^2}{\left(\rho^2V_0+\vm\right)^2} \\
            % &= \frac{1}{V_0^2}-\frac{\delta\rho^2}{\delta^{2/3}\rho^{8/3}(V_0)^2} \\
            &= \frac{1-\left(\delta/\rho^2\right)^{1/3}}{(V_0)^2}
        \end{align*}
        is strictly negative because~$\delta>\rho^2$.
        So as~$V$ rises, the value~$f(V)$ drops from~$\infty$ to every positive real number (including~$1/c$) before reaching its negative minimum, then stays negative while rising toward zero.
        Thus~$f(V)=1/c>0$ is satisfied at most once.

    \end{enumerate}
    Together, cases~(i) and~(ii) imply~$V^*>0$ exists and is unique.
    They also imply~$f'(V^*)<0$.

    Now I derive the agent's optimal choice of~$V_t$.
    The minimand on the right-hand side of~\eqref{eq:bellman} is differentiable in~$P_t$, implying~$\pv$ is also differentiable.
    Denoting its derivative by~$\pv'$, assuming an interior minimizer, and applying the envelope theorem gives
    \begin{align*}
        \pv'(P_t)
        &= \parfrac{}{P_t}\left\{\cost(V_t,P_t)+\delta\pv(\rho^2V_t+\vm)\right\} \\
        &= \frac{c}{P_t^2}.
    \end{align*}
    So any interior minimizer satisfies the first-order condition
    \begin{align*}
        0
        &= \parfrac{}{V_t}\left\{\cost(V_t,P_t)+\delta\pv(\rho^2V_t+\vm)\right\} \\
        &= 1-\frac{c}{V_t^2}+\delta\rho^2\pv'(\rho^2V_t+\vm) \\
        % &= 1-c\left(\frac{1}{V_t^2}-\frac{\delta\rho^2}{\rho^2V_t+\vm}\right) \\
        &= c\left(\frac{1}{c}-f(V_t)\right),
    \end{align*}
    which is uniquely satisfied by~$V^*$.
    The second-order condition for a minimizer holds because
    \[ \parfrac{^2}{V_t^2}\left\{\cost(V_t,P_t)+\delta\pv(\rho^2V_t+\vm)\right\}\bigg\vert_{V_t=V^*}=-f'(V^*) \]
    is strictly positive.
    So the agent optimally chooses~$V_t=V^*$ when~$V^*<P_t$ and~$V_t=P_t$ otherwise, yielding the optimal policy~$V_t=\min\{P_t,V^*\}$.
    The expression for~$x_t=1/V_t-1/P_t$ follows.
\end{proof}
\fi

I prove Lemma~\ref{lem:optimal-policy} and all other results in Appendix~\ref{sec:proofs}.

\iffalse
For example, suppose~$(\rho,\vz,\vm,c,\delta)=(0.5,0,1,1,0)$.
Then~\eqref{eq:foc} has unique solution~$V^*=1$.
At time~$t=1$, the state has prediction variance~$P_1=\rho^2\vz+\vm=1$ and so the agent chooses~$V_1=\min\{1,1\}=1$ by generating a signal with precision~$x_1=0$.
At time~$t=2$, the prediction variance is~$P_2=\rho^2V_1+\vm=1.25$ and he chooses~$V_2=\min\{1.25,1\}=1$ via~$x_2=\max\{1/1-1/1.25,0\}=0.2$.
At time~$t=3$, the prediction variance is~$P_3=1.25$ and he chooses~$V_3=1$ via~$x_3=0.2$ again.
In fact he chooses~$V_t=1$ via~$x_t=0.2$ at each time~$t\ge2$; that is, his optimal strategy converges to a steady state.
I provide sufficient conditions for such convergence in Section~\ref{sec:steady-state}.
\fi

Suppose the agent never acquires information: $x_t=0$ for each~$t$.
Then Lemma~\ref{lem:optimal-policy} implies~$V_t=P_t$ for each~$t$, and the AR(1) structure implies~$P_1=\rho^2\vz+\vm$ and~$P_t=\rho^2P_{t-1}+\vm$ for each~$t>1$.
It follows that
\begin{equation}
    \label{eq:prediction-variance-uninformed}
    P_t=\rho^{2t}\vz+\left(\frac{1-\rho^{2t}}{1-\rho^2}\right)\vm
\end{equation}
for each~$t\in\N$, which converges to~$\vm/(1-\rho^2)$ as~$t\to\infty$.
But if
\begin{equation}
    \label{eq:steady-state-condition}
    V^*<\frac{\vm}{1-\rho^2},
\end{equation}
then the sequence~$(P_t)_{t\in\N}$ defined by~\eqref{eq:prediction-variance-uninformed} eventually exceeds~$V^*$, in which case Lemma~\ref{lem:optimal-policy} implies the agent would optimally choose~$V_t=V^*$.

So if~\eqref{eq:steady-state-condition} holds, then there is a least time~$t^*\in\N$ such that the agent optimally chooses~$V_{t^*}=V^*$.
Moreover, the next prediction variance~$P_{t^*+1}=\rho^2V^*+\vm$ exceeds~$V^*$\iffalse (since~\eqref{eq:steady-state-condition} holds if and only if~$V^*<\rho^2V^*+\vm)$\fi, implying~$V_{t^*+1}=V^*$ by Lemma~\ref{lem:optimal-policy} and~$V_t=V^*$ for each~$t\ge t^*$ by induction.
Conversely, if~\eqref{eq:steady-state-condition} does not hold, then~$P_t$ may never fall below~$V^*$ and the agent may optimally never acquire information.

Lemma~\ref{lem:steady-state-condition-primitives} says that~\eqref{eq:steady-state-condition} holds precisely when the marginal precision cost~$c$ is small.
If it is too large, then information is not worth acquiring and the agent optimally remains uninformed.

\begin{lemma}
    \label{lem:steady-state-condition-primitives}
    Define~$V^*$ as in Lemma~\ref{lem:optimal-policy}.
    Then~$V^*<\vm/(1-\rho^2)$ if and only if
    \begin{equation}
        \label{eq:steady-state-condition-primitives}
        c<\frac{\sigma^4}{(1-\delta\rho^2)(1-\rho^2)^2}.
    \end{equation}
\end{lemma}
\ifbodyproofs\subsection{Proof of Lemma~\ref{lem:steady-state-condition-primitives}}

\begin{proof}[\nopunct\unskip]
    Let~$V_0\equiv\vm/(1-\rho^2)$ and let~$f$ be the function defined in the proof of Lemma~\ref{lem:optimal-policy}.
    Then~$f(V)$ is strictly decreasing in~$V>0$ whenever~$f(V)>0$.
    Since~$1/c>0$, it follows that~$V^*<V_0$ if and only if~$f(V_0)>1/c$.
    But~$\rho^2V_0+\vm=V_0$, and so
    \begin{align*}
        f(V_0)
        &= \frac{1}{V_0^2}-\frac{\delta\rho^2}{\left(\rho^2V_0+\vm\right)^2} \\
        % &= \frac{1-\delta\rho^2}{V_0^2} \\
        &= \frac{\left(1-\delta\rho^2\right)\left(1-\rho^2\right)^2}{\sigma^4}
    \end{align*}
    exceeds~$1/c$ if and only if~\eqref{eq:steady-state-condition-primitives} holds.
\end{proof}
\fi

\section{Steady-state analysis}

Suppose~\eqref{eq:steady-state-condition-primitives} holds.
Then, by Lemmas~\ref{lem:optimal-policy} and~\ref{lem:steady-state-condition-primitives}, the agent's optimal policy reaches a steady state in which he sets the posterior variance~$V_t$ equal to the target~$V^*$ defined by~\eqref{eq:foc}.
He maintains~$V_t=V^*$ indefinitely by generating signals with precision%
\footnote{
We have~$x^*>0$ if and only if~$V^*<\rho^2V^*+\vm$, which, by Lemma~\ref{lem:steady-state-condition-primitives}, holds if and only if~\eqref{eq:steady-state-condition-primitives} holds.
}
\begin{equation}
    \label{eq:steady-state-precision}
    x^*\equiv\frac{1}{V^*}-\frac{1}{\rho^2V^*+\vm}.
\end{equation}
I call~$V^*$ the ``steady-state variance'' and~$x^*$ the ``steady-state precision.''

\subsection{How $V^*$ and~$x^*$ depend on primitives}

\begin{figure}
    \centering
    \includegraphics[width=0.75\linewidth]{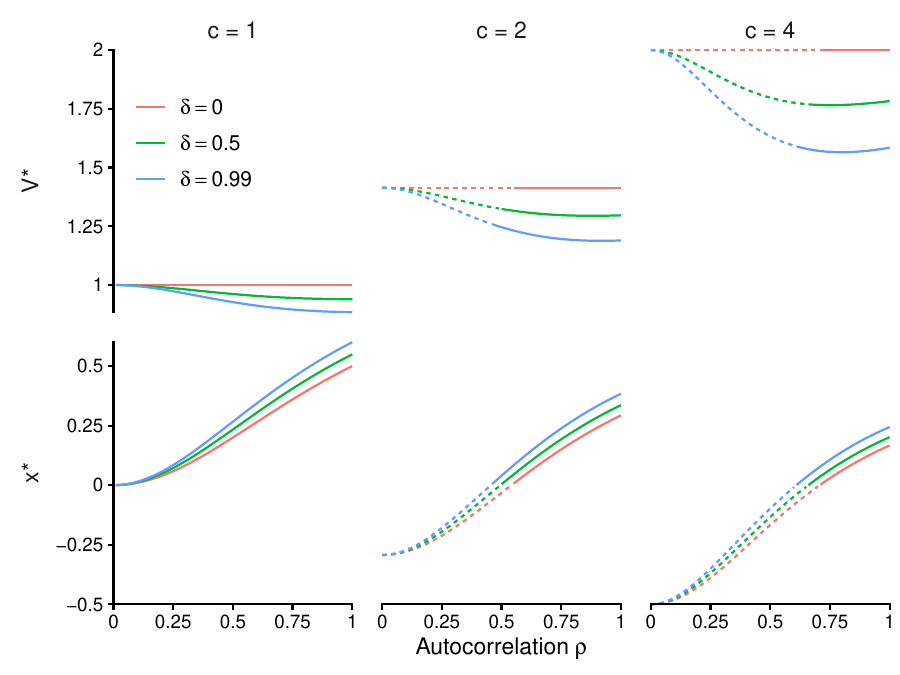}
    \caption{Steady-state variance~$V^*$ and precision~$x^*$ when~$\vm=1$}
    \label{fig:steady-state}
    \caption*{
    \footnotesize
    {\itshape Note:}
    Solid lines are values where~\eqref{eq:steady-state-condition-primitives} holds; dashed lines are values where~\eqref{eq:steady-state-condition-primitives} does not hold.
    }
\end{figure}
Figure~\ref{fig:steady-state} shows that~$V^*$ can rise or fall when the autocorrelation~$\rho$ rises.
If~$\rho$ is small, then most of the state's variation over time comes from the independent shocks, and so signals become obsolete quickly.
As~$\rho$ rises, signals become obsolete slower, and so raising their precision leads to a larger decrease in future instrumental costs.
The agent responds by generating more precise signals, pushing~$V^*$ down.
An opposing force pushes~$V^*$ up: if~$\rho$ rises, then the state inherits more variance from past shocks.
Proposition~\ref{prop:steady-state-rho} says this opposing force dominates precisely when~$\rho$ is large.

\begin{proposition}
    \label{prop:steady-state-rho}
    Let~$\delta>0$ and define
    \[ \rho^*\equiv\sqrt{\frac{\delta}{8}+\sqrt{\frac{\delta^2}{64}+\frac{\sigma^4}{c}}}. \]
    If~$\rho<\rho^*$, then the steady-state variance~$V^*$ is decreasing in the autocorrelation~$\rho$.
    If~$\rho>\rho^*$, then~$V^*$ is increasing in~$\rho$.
\end{proposition}
\ifbodyproofs\subsection{Proof of Proposition~\ref{prop:steady-state-rho}}
\begin{proof}[\unskip\nopunct]
    Differentiating~\eqref{eq:foc} implicitly with respect to~$\rho$ and rearranging the result gives
    \[ \left(\frac{1}{(V^*)^3}-\frac{\delta\rho^4}{\left(\rho^2V^*+\vm\right)^3}\right)\derfrac{V^*}{\rho}=\frac{\delta\rho\left(\rho^2V^*-\vm\right)}{\left(\rho^2V^*+\vm\right)^3}. \]
    The bracketed term on the left-hand side equals~$-f'(V^*)/2$, where~$f$ is the function defined in the proof of Lemma~\ref{lem:optimal-policy}.
    Part of that proof showed~$f'(V^*)<0$.
    So the bracketed term is positive, and hence~$\der V^*/\der\rho<0$ if and only if
    \begin{equation}
        \label{eq:steady-state-rho-condition-proof}
        \rho^2<\frac{\vm}{V^*}
    \end{equation}
    But~$V^*$ is endogenous: it depends on the primitives~$\delta$, $\rho$, $\vm$, and~$c$.
    To express the right-hand side of~\eqref{eq:steady-state-rho-condition-proof} in terms of these primitives, substitute~$V^*=\vm/\rho^2$ into~\eqref{eq:foc} to obtain
    \begin{align*}
        \frac{1}{c}
        &= \frac{1}{(\vm/\rho^2)^2}-\frac{\delta\rho^2}{\left(\rho^2(\vm/\rho^2)+\vm\right)^2} \\
        &= \frac{(4\rho^2-\delta)\rho^2}{4\sigma^4}.
    \end{align*}
    Rearranging this equation yields a quadratic
    \[ 0=4\rho^4-\delta\rho^2-\frac{4\sigma^4}{c} \]
    in~$\rho^2$, with roots
    \[ \rho^2=\frac{\delta}{8}\pm\sqrt{\frac{\delta^2}{64}+\frac{\sigma^4}{c}}. \]
    The negative root is infeasible; rejecting it yields the result.
\end{proof}
\fi

\iffalse
For example, suppose~$(\vm,c,\delta)=(1,4,0.99)$.
Then~$\rho^*\approx0.8$, which is where the bottom curve in the top-right panel of Figure~\ref{fig:steady-state} attains its minimum.
\fi

If~$\delta>0$, then the agent has an incentive to lower future instrumental costs by raising current acquisition costs.
Letting~$\delta=0$ removes this incentive and makes~$V^*=\sqrt{c}$ constant in~$\rho$.
Raising the discount factor~$\delta$ makes the agent care more about his future instrumental costs.
He lowers these costs by learning more about the state's current value, lowering~$V^*$ by raising~$x^*$.
Figure~\ref{fig:steady-state} illustrates this behavior for specific parameter values; Proposition~\ref{prop:steady-state-delta} establishes the behavior for \emph{all} parameter values.

\begin{proposition}
    \label{prop:steady-state-delta}
    Suppose~\eqref{eq:steady-state-condition-primitives} holds.
    Then the steady-state variance~$V^*$ is decreasing in the discount factor~$\delta$, whereas the steady-state precision~$x^*$ is increasing in~$\delta$.
\end{proposition}
\ifbodyproofs\subsection{Proof of Proposition~\ref{prop:steady-state-delta}}

\begin{proof}[\unskip\nopunct]
    Differentiating~\eqref{eq:foc} implicitly with respect to~$\delta$ and rearranging the result gives
    \begin{equation}
        \label{eq:dV/ddelta}
        -2\left(\frac{1}{(V^*)^3}-\frac{\delta\rho^4}{\left(\rho^2V^*+\vm\right)^3}\right)\derfrac{V^*}{\delta}=\frac{\rho^2}{\left(\rho^2V^*+\vm\right)^2}.
    \end{equation}
    The term multiplying~$\der V^*/\der\delta$ on the left-hand side equals~$f'(V^*)$, where~$f$ is the function defined in the proof of Lemma~\ref{lem:optimal-policy}.
    Part of that proof showed~$f'(V^*)<0$.
    So~$\der V^*/\der\delta$ has the opposite sign to the right-hand side of~\eqref{eq:dV/ddelta}, which is strictly positive.
    Thus~$V^*$ is decreasing in~$\delta$.

    Differentiating~\eqref{eq:steady-state-precision} with respect to~$\delta$ gives
    \[ \derfrac{x^*}{\delta}=-\left(\frac{1}{(V^*)^2}-\frac{\rho^2}{\left(\rho^2V^*+\vm\right)^2}\right)\derfrac{V^*}{\delta}. \]
    But Lemma~\ref{lem:steady-state-condition-primitives} implies~$V^*<\rho^2V^*+\vm$.
    Also~$\rho^2\in(0,1)$.
    Thus
    \[ \frac{1}{(V^*)^2}>\frac{\rho^2}{\left(\rho^2V^*+\vm\right)^2}. \]
    It follows that~$\der x^*/\der\delta$ and~$\der V^*/\der\delta$ have opposite signs, and so~$x^*$ is increasing in~$\delta$.
\end{proof}
\fi

Figure~\ref{fig:steady-state} also shows that~$V^*$ and~$x^*$ depend on the marginal precision cost~$c$.
Raising~$c$ makes information more expensive, so the agent acquires less and maintains less precise beliefs:

\begin{proposition}
    \label{prop:steady-state-c}
    Suppose~\eqref{eq:steady-state-condition-primitives} holds.
    Then the steady-state variance~$V^*$ is increasing in the marginal precision cost~$c$, whereas the steady-state precision~$x^*$ is decreasing in~$c$.
\end{proposition}
\ifbodyproofs\subsection{Proof of Proposition~\ref{prop:steady-state-c}}

\begin{proof}[\unskip\nopunct]
    Differentiating~\eqref{eq:foc} implicitly with respect to~$c$ gives
    \begin{equation}
        \label{eq:dV/dc}
        -2\left(\frac{1}{(V^*)^3}-\frac{\delta\rho^4}{\left(\rho^2V^*+\vm\right)^3}\right)\derfrac{V^*}{c}=-\frac{1}{c^2}.
    \end{equation}
    From the proof of Proposition~\ref{prop:steady-state-delta}, we know the term multiplying~$\der V^*/\der c$ on the left-hand side is strictly negative.
    So~$\der V^*/\der c$ has the opposite sign to the right-hand side of~\eqref{eq:dV/dc}, which is strictly negative.
    Thus~$V^*$ is increasing in~$c$.

    Differentiating~$x^*$ with respect to~$c$ gives
    \[ \derfrac{x^*}{c}=-\left(\frac{1}{(V^*)^2}-\frac{\rho^2}{\left(\rho^2V^*+\vm\right)^2}\right)\derfrac{V^*}{c}. \]
    From the proof of Proposition~\ref{prop:steady-state-delta}, we know the term multiplying~$\der V^*/\der c$ on the right-hand side is strictly negative.
    Thus~$\der x^*/\der c$ and~$\der V^*/\der c$ have opposite signs, and so~$x^*$ is decreasing in~$c$.
\end{proof}
\fi

\subsection{Steady-state cost}

Finally, I define the ``steady-state cost''
\begin{align*}
    \cost^*
    &\equiv \cost(V^*,\rho^2V^*+\vm) \\
    &= V^*+cx^*
\end{align*}
as the overall cost the agent incurs at each time in steady state.
He is better off when~$\cost^*$ is lower.
\iffalse
\footnote{
Indeed, substituting~$P_t=P^*\equiv\rho^2V^*+\vm$ into the Bellman equation~\eqref{eq:bellman} gives
%
\begin{align*}
    \pv(P^*)
    &= \min_{V\in(0,P^*]}\left\{\cost(V,P^*)+\delta\pv(\rho^2V+\vm)\right\} \\
    &= \cost^*+\delta\pv(P^*).
\end{align*}
%
So~$\cost^*=(1-\delta)\pv(P^*)$ is proportional to the present value of the agent's overall costs in the steady state.
% Intuitively, since the optimal policy is stationary in steady state, the normalized continuation value equals the per-period cost.
}
\fi
The non-monotonicity of~$V^*$ in~$\rho$ suggests~$\cost^*$ could also be non-monotone.
Likewise, since~$V^*$ and~$x^*$ move in different directions when~$\delta$ or~$c$ rise, the net effect on~$\cost^*$ is potentially ambiguous.
Proposition~\ref{prop:steady-state-cost} says the effects are, in fact, unambiguous: $\cost^*$ is monotone in \emph{all} the primitives that determine~$V^*$.

\begin{proposition}
    \label{prop:steady-state-cost}
    Suppose~\eqref{eq:steady-state-condition-primitives} holds.
    Then the steady-state cost~$\cost^*$ is
    \begin{enumerate}
        
        \item[(i)]
        increasing in the autocorrelation~$\rho$,
        
        \item[(ii)]
        decreasing in the discount factor~$\delta$,

        \item[(iii)]
        increasing in the marginal precision cost~$c$, and

        \item[(iv)]
        increasing in the shock variance~$\vm$.

    \end{enumerate}
\end{proposition}
\ifbodyproofs\subsection{Proof of Proposition~\ref{prop:steady-state-cost}}

\begin{proof}[\unskip\nopunct]
    Define~$P^*\equiv\rho^2V^*+\vm>0$ for convenience.
    I establish~(i)--(iv) separately:
    \begin{enumerate}

        \item[(i)]
        Differentiating~$P^*$ with respect to~$\rho$ gives
        \[ \derfrac{P^*}{\rho}=2\rho V^*+\rho^2\derfrac{V^*}{\rho} \]
        and so
        \[ \derfrac{V^*}{\rho}=\frac{1}{\rho^2}\left(\derfrac{P^*}{\rho}-2\rho V^*\right). \]
        We also have
        \[ \derfrac{x^*}{\rho}=-\frac{1}{(V^*)^2}\derfrac{V^*}{\rho}+\frac{1}{(P^*)^2}\derfrac{P^*}{\rho}. \]
        Thus
        \begin{align*}
            \derfrac{\cost^*}{\rho}
            &= \derfrac{V^*}{\rho}+c\derfrac{x^*}{\rho} \\
            &= \frac{1}{\rho^2}\left(\derfrac{P^*}{\rho}-2\rho V^*\right)+c\left(-\frac{1}{(V^*)^2}\left(\frac{1}{\rho^2}\left(\derfrac{P^*}{\rho}-2\rho V^*\right)\right)+\frac{1}{(P^*)^2}\derfrac{P^*}{\rho}\right) \\
            % &= \frac{1}{\rho^2}\left(\derfrac{P^*}{\rho}-2\rho V^*-\frac{c}{(V^*)^2}\left(\derfrac{P^*}{\rho}-2\rho V^*\right)+\frac{c\rho^2}{(P^*)^2}\derfrac{P^*}{\rho}\right) \\
            &= \frac{c}{\rho^2}\left(\left(\frac{1}{c}-\left(\frac{1}{(V^*)^2}-\frac{\rho^2}{(P^*)^2}\right)\right)\derfrac{P^*}{\rho}-2\rho V^*\left(\frac{1}{c}-\frac{1}{(V^*)^2}\right)\right).
        \end{align*}
        Substituting in the definition~\eqref{eq:foc} of~$V^*$ and simplifying the result gives
        \[ \derfrac{\cost^*}{\rho}=\frac{c}{(P^*)^2}\left((1-\delta)\derfrac{P^*}{\rho}+2\rho\delta V^*\right). \]
        But~$c/(P^*)^2$ and $(1-\delta)$ are strictly positive, and~$2\rho\delta V^*$ is non-negative.
        Therefore, it suffices to show~$\der P^*/\der\rho$ is strictly positive.
        Substituting~$V^*\equiv(P^*-\vm)/\rho^2$ into~\eqref{eq:foc} gives
        \[ \frac{1}{c}=\frac{\rho^4}{(P^*-\vm)^2}-\frac{\delta\rho^2}{(P^*)^2}. \]
        Differentiating implicitly with respect to~$\rho$ and rearranging the result gives
        \[ -2\left(\frac{\rho^4}{\left(P^*-\vm\right)^3}-\frac{\delta\rho^2}{(P^*)^3}\right)\derfrac{P^*}{\rho}=-2\left(\frac{2\rho^3}{\left(P^*-\vm\right)^2}-\frac{\delta\rho}{(P^*)^2}\right), \]
        which we can also write as
        \begin{equation}
            \label{eq:dP*/drho}
            -\frac{2}{\rho^2}\left(\frac{1}{(V^*)^3}-\frac{\delta\rho^4}{(P^*)^3}\right)\derfrac{P^*}{\rho}=-\frac{2}{\rho}\left(\frac{2}{(V^*)^2}-\frac{\delta\rho^2}{(P^*)^2}\right).
        \end{equation}
        Now consider the function~$f$ defined in the proof of Lemma~\ref{lem:optimal-policy}.
        Part of that proof showed
        \[ f'(V^*)=-2\left(\frac{1}{(V^*)^3}-\frac{\delta\rho^4}{(P^*)^3}\right) \]
        is strictly negative.
        Moreover, the definitions of~$V^*$ and~$P^*$ imply
        \[ \frac{\delta\rho^2}{(P^*)^2}=\frac{1}{(V^*)^2}-\frac{1}{c} \]
        Substituting these expressions into~\eqref{eq:dP*/drho} and rearranging the result gives
        \[ \derfrac{P^*}{\rho}=-\frac{2\rho}{f'(V^*)}\left(\frac{1}{(V^*)^2}+\frac{1}{c}\right). \]
        The right-hand side is strictly positive, so~$P^*$ (and thus~$\cost^*$) is increasing in~$\rho$.

        \item[(ii)]
        From the proof of Proposition~\ref{prop:steady-state-delta}, we know~$\der V^*/\der\delta$ is strictly negative and
        \[ \derfrac{x^*}{\delta}=-\left(\frac{1}{(V^*)^2}-\frac{\rho^2}{(P^*)^2}\right)\derfrac{V^*}{\delta}. \]
        So differentiating~$\cost^*$ with respect to~$\delta$ gives
        \begin{align*}
            \derfrac{\cost^*}{\delta}
            &= \derfrac{V^*}{\delta}+c\derfrac{x^*}{\delta} \\
            &= c\left(\frac{1}{c}-\left(\frac{1}{(V^*)^2}-\frac{\rho^2}{(P^*)^2}\right)\right)\derfrac{V^*}{\delta} \\
            % &= c\left(\left(\frac{1}{(V^*)^2}-\frac{\delta\rho^2}{(P^*)^2}\right)-\left(\frac{1}{(V^*)^2}-\frac{\rho^2}{(P^*)^2}\right)\right)\derfrac{V^*}{\delta} \\
            &= \frac{c(1-\delta)\rho^2}{(P^*)^2}\derfrac{V^*}{\delta},
        \end{align*}
        where the last equality follows from the definition~\eqref{eq:foc} of~$V^*$.
        So~$\der\cost^*/\der \delta$ and~$\der V^*/\der\delta$ have the same sign, and hence~$\cost^*$ is decreasing in~$\delta$.

        \item[(iii)]
        From the proof of Proposition~\ref{prop:steady-state-c}, we know~$\der V^*/\der c$ is strictly positive and
        \[ \derfrac{x^*}{c}=-\left(\frac{1}{(V^*)^2}-\frac{\rho^2}{(P^*)^2}\right)\derfrac{V^*}{c}. \]
        So differentiating~$\cost^*$ with respect to~$c$ gives
        \begin{align*}
            \derfrac{\cost^*}{c}
            &= \derfrac{V^*}{c}+x^*+c\derfrac{x^*}{c} \\
            &= c\left(\frac{1}{c}-\left(\frac{1}{(V^*)^2}-\frac{\rho^2}{(P^*)^2}\right)\right)\derfrac{V^*}{c}+x^* \\
            &= \frac{c(1-\delta)\rho^2}{(P^*)^2}\derfrac{V^*}{c}+x^*,
        \end{align*}
        where the last equality follows from the definition~\eqref{eq:foc} of~$V^*$.
        All terms on the right-hand side are strictly positive, so~$\cost^*$ is increasing in~$c$.

        \item[(iv)]
        Differentiating~$P^*$ with respect to~$\vm$ gives
        \[ \derfrac{P^*}{\vm}=\rho^2\derfrac{V^*}{\vm}+1 \]
        and so
        \[ \derfrac{V^*}{\vm}=\frac{1}{\rho^2}\left(\derfrac{P^*}{\vm}-1\right). \]
        Repeating the argument used in the proof of part~(i) gives
        \begin{align*}
            \derfrac{\cost^*}{\vm}
            &= \frac{c}{(P^*)^2}\left((1-\delta)\derfrac{P^*}{\vm}+\delta\right) \\
            &= \frac{c}{(P^*)^2}\left((1-\delta)\rho^2\derfrac{V^*}{\vm}+1\right).
        \end{align*}
        Now~$c/(P^*)^2$ and~$(1-\delta)\rho^2$ are strictly positive, and so it suffices to show~$\der V^*/\der\vm$ is strictly positive.
        Differentiating~\eqref{eq:foc} with respect to~$\vm$ and rearranging the result gives
        \[ -f'(V^*)\derfrac{V^*}{\vm}=\frac{2\delta\rho^2}{(P^*)^3}, \]
        where~$f$ is the function defined in the proof of Lemma~\ref{lem:optimal-policy}.
        But~$f'(V^*)$ is strictly negative and the right-hand side is strictly positive, so~$\der V^*/\der\vm$ must also be strictly positive.
        Thus~$V^*$ (and therefore~$\cost^*$) is increasing in~$\vm$.
        \qedhere

    \end{enumerate}
\end{proof}
\fi

Part~(i) of Proposition~\ref{prop:steady-state-cost} says the agent is worse off when the state is more persistent.
This is because he optimally acquires more information, and the acquisition cost~$cx^*$ rises by more than the instrumental cost~$V^*$ falls.
Part~(ii) says he is better off when he is more patient.
This is because his past selves acquire more information, lowering~$V^*$ enough to offset the rise in~$cx^*$.
Parts~(iii) and~(iv) say he is worse off when information is more expensive and when shocks have higher variance.
This is because he has to buy more information to maintain a given posterior variance.

{
\raggedright
\bibliographystyle{apalike}
\bibliography{references}

\begin{thebibliography}{}

\bibitem[Barilla, 2025]{Barilla-2025-}
Barilla, C. (2025).
\newblock When and what to learn in a changing world.
\newblock {{arXiv}} preprint 2510.17757.

\bibitem[Colombo et~al., 2014]{Colombo-etal-2014-REStud}
Colombo, L., Femminis, G., and Pavan, A. (2014).
\newblock Information {{Acquisition}} and {{Welfare}}.
\newblock {\em Review of Economic Studies}, 81(4):1438--1483.

\bibitem[Davies, 2026]{Davies-2026-}
Davies, B. (2026).
\newblock Learning about a changing state.
\newblock {{arXiv}} preprint 2401.03607.

\bibitem[Hellwig and Veldkamp, 2009]{Hellwig-Veldkamp-2009-REStud}
Hellwig, C. and Veldkamp, L. (2009).
\newblock Knowing {{What Others Know}}: {{Coordination Motives}} in {{Information Acquisition}}.
\newblock {\em Review of Economic Studies}, 76(1):223--251.

\bibitem[Immorlica et~al., 2021]{Immorlica-etal-2021-LeibnizInt.Proc.Inform.LIPIcs}
Immorlica, N., Kash, I.~A., and Lucier, B. (2021).
\newblock Buying {{Data Over Time}}: {{Approximately Optimal Strategies}} for {{Dynamic Data-Driven Decisions}}.
\newblock In {\em Leibniz {{International Proceedings}} in {{Informatics}} ({{LIPIcs}})}, volume 185, pages 77:1--77:14, Dagstuhl, Germany.

\bibitem[Ma{\'c}kowiak and Wiederholt, 2009]{Mackowiak-Wiederholt-2009-AER}
Ma{\'c}kowiak, B. and Wiederholt, M. (2009).
\newblock Optimal {{Sticky Prices}} under {{Rational Inattention}}.
\newblock {\em American Economic Review}, 99(3):769--803.

\bibitem[Myatt and Wallace, 2012]{Myatt-Wallace-2012-REStud}
Myatt, D.~P. and Wallace, C. (2012).
\newblock Endogenous {{Information Acquisition}} in {{Coordination Games}}.
\newblock {\em Review of Economic Studies}, 79(1):340--374.

\bibitem[Pomatto et~al., 2023]{Pomatto-etal-2023-AER}
Pomatto, L., Strack, P., and Tamuz, O. (2023).
\newblock The {{Cost}} of {{Information}}: {{The Case}} of {{Constant Marginal Costs}}.
\newblock {\em American Economic Review}, 113(5):1360--1393.

\bibitem[Sims, 2003]{Sims-2003-JME}
Sims, C.~A. (2003).
\newblock Implications of rational inattention.
\newblock {\em Journal of Monetary Economics}, 50(3):665--690.

\bibitem[Steiner et~al., 2017]{Steiner-etal-2017-ECTA}
Steiner, J., Stewart, C., and Mat{\v e}jka, F. (2017).
\newblock Rational {{Inattention Dynamics}}: {{Inertia}} and {{Delay}} in {{Decision-Making}}.
\newblock {\em Econometrica}, 85(2):521--553.

\bibitem[Weber and Nguyen, 2018]{Weber-Nguyen-2018-EJOR}
Weber, T.~A. and Nguyen, V.~A. (2018).
\newblock A linear-quadratic {{Gaussian}} approach to dynamic information acquisition.
\newblock {\em European Journal of Operational Research}, 270(1):260--281.

\end{thebibliography}
}

\appendix

\clearpage
\section{Proofs}
\label{sec:proofs}

\counterwithin{equation}{section}
\counterwithin{lemma}{section}
\counterwithin{proposition}{section}
\setcounter{equation}{0}
\setcounter{lemma}{0}
\setcounter{proposition}{0}

\ifbodyproofs\else\fi

\ifbodyproofs\else\subsection{Proof of Lemma~\ref{lem:steady-state-condition-primitives}}

\begin{proof}[\nopunct\unskip]
    Let~$V_0\equiv\vm/(1-\rho^2)$ and let~$f$ be the function defined in the proof of Lemma~\ref{lem:optimal-policy}.
    Then~$f(V)$ is strictly decreasing in~$V>0$ whenever~$f(V)>0$.
    Since~$1/c>0$, it follows that~$V^*<V_0$ if and only if~$f(V_0)>1/c$.
    But~$\rho^2V_0+\vm=V_0$, and so
    \begin{align*}
        f(V_0)
        &= \frac{1}{V_0^2}-\frac{\delta\rho^2}{\left(\rho^2V_0+\vm\right)^2} \\
        % &= \frac{1-\delta\rho^2}{V_0^2} \\
        &= \frac{\left(1-\delta\rho^2\right)\left(1-\rho^2\right)^2}{\sigma^4}
    \end{align*}
    exceeds~$1/c$ if and only if~\eqref{eq:steady-state-condition-primitives} holds.
\end{proof}
\fi

\ifbodyproofs\else\subsection{Proof of Proposition~\ref{prop:steady-state-rho}}
\begin{proof}[\unskip\nopunct]
    Differentiating~\eqref{eq:foc} implicitly with respect to~$\rho$ and rearranging the result gives
    \[ \left(\frac{1}{(V^*)^3}-\frac{\delta\rho^4}{\left(\rho^2V^*+\vm\right)^3}\right)\derfrac{V^*}{\rho}=\frac{\delta\rho\left(\rho^2V^*-\vm\right)}{\left(\rho^2V^*+\vm\right)^3}. \]
    The bracketed term on the left-hand side equals~$-f'(V^*)/2$, where~$f$ is the function defined in the proof of Lemma~\ref{lem:optimal-policy}.
    Part of that proof showed~$f'(V^*)<0$.
    So the bracketed term is positive, and hence~$\der V^*/\der\rho<0$ if and only if
    \begin{equation}
        \label{eq:steady-state-rho-condition-proof}
        \rho^2<\frac{\vm}{V^*}
    \end{equation}
    But~$V^*$ is endogenous: it depends on the primitives~$\delta$, $\rho$, $\vm$, and~$c$.
    To express the right-hand side of~\eqref{eq:steady-state-rho-condition-proof} in terms of these primitives, substitute~$V^*=\vm/\rho^2$ into~\eqref{eq:foc} to obtain
    \begin{align*}
        \frac{1}{c}
        &= \frac{1}{(\vm/\rho^2)^2}-\frac{\delta\rho^2}{\left(\rho^2(\vm/\rho^2)+\vm\right)^2} \\
        &= \frac{(4\rho^2-\delta)\rho^2}{4\sigma^4}.
    \end{align*}
    Rearranging this equation yields a quadratic
    \[ 0=4\rho^4-\delta\rho^2-\frac{4\sigma^4}{c} \]
    in~$\rho^2$, with roots
    \[ \rho^2=\frac{\delta}{8}\pm\sqrt{\frac{\delta^2}{64}+\frac{\sigma^4}{c}}. \]
    The negative root is infeasible; rejecting it yields the result.
\end{proof}
\fi

\ifbodyproofs\else\subsection{Proof of Proposition~\ref{prop:steady-state-delta}}

\begin{proof}[\unskip\nopunct]
    Differentiating~\eqref{eq:foc} implicitly with respect to~$\delta$ and rearranging the result gives
    \begin{equation}
        \label{eq:dV/ddelta}
        -2\left(\frac{1}{(V^*)^3}-\frac{\delta\rho^4}{\left(\rho^2V^*+\vm\right)^3}\right)\derfrac{V^*}{\delta}=\frac{\rho^2}{\left(\rho^2V^*+\vm\right)^2}.
    \end{equation}
    The term multiplying~$\der V^*/\der\delta$ on the left-hand side equals~$f'(V^*)$, where~$f$ is the function defined in the proof of Lemma~\ref{lem:optimal-policy}.
    Part of that proof showed~$f'(V^*)<0$.
    So~$\der V^*/\der\delta$ has the opposite sign to the right-hand side of~\eqref{eq:dV/ddelta}, which is strictly positive.
    Thus~$V^*$ is decreasing in~$\delta$.

    Differentiating~\eqref{eq:steady-state-precision} with respect to~$\delta$ gives
    \[ \derfrac{x^*}{\delta}=-\left(\frac{1}{(V^*)^2}-\frac{\rho^2}{\left(\rho^2V^*+\vm\right)^2}\right)\derfrac{V^*}{\delta}. \]
    But Lemma~\ref{lem:steady-state-condition-primitives} implies~$V^*<\rho^2V^*+\vm$.
    Also~$\rho^2\in(0,1)$.
    Thus
    \[ \frac{1}{(V^*)^2}>\frac{\rho^2}{\left(\rho^2V^*+\vm\right)^2}. \]
    It follows that~$\der x^*/\der\delta$ and~$\der V^*/\der\delta$ have opposite signs, and so~$x^*$ is increasing in~$\delta$.
\end{proof}
\fi

\ifbodyproofs\else\subsection{Proof of Proposition~\ref{prop:steady-state-c}}

\begin{proof}[\unskip\nopunct]
    Differentiating~\eqref{eq:foc} implicitly with respect to~$c$ gives
    \begin{equation}
        \label{eq:dV/dc}
        -2\left(\frac{1}{(V^*)^3}-\frac{\delta\rho^4}{\left(\rho^2V^*+\vm\right)^3}\right)\derfrac{V^*}{c}=-\frac{1}{c^2}.
    \end{equation}
    From the proof of Proposition~\ref{prop:steady-state-delta}, we know the term multiplying~$\der V^*/\der c$ on the left-hand side is strictly negative.
    So~$\der V^*/\der c$ has the opposite sign to the right-hand side of~\eqref{eq:dV/dc}, which is strictly negative.
    Thus~$V^*$ is increasing in~$c$.

    Differentiating~$x^*$ with respect to~$c$ gives
    \[ \derfrac{x^*}{c}=-\left(\frac{1}{(V^*)^2}-\frac{\rho^2}{\left(\rho^2V^*+\vm\right)^2}\right)\derfrac{V^*}{c}. \]
    From the proof of Proposition~\ref{prop:steady-state-delta}, we know the term multiplying~$\der V^*/\der c$ on the right-hand side is strictly negative.
    Thus~$\der x^*/\der c$ and~$\der V^*/\der c$ have opposite signs, and so~$x^*$ is decreasing in~$c$.
\end{proof}
\fi

\ifbodyproofs\else\fi

\end{document}